\pdfoutput=1

\documentclass[11pt,letterpaper]{article}

\usepackage[utf8]{inputenc}
\usepackage[T1]{fontenc}
\usepackage[margin=1in]{geometry}
\usepackage{url}
\usepackage{ifthen}
\usepackage{cite}
\usepackage{amsmath}
\usepackage{amssymb}
\usepackage{amsthm}
\usepackage{graphicx}
\usepackage{enumitem}
\usepackage{tikz}

\theoremstyle{plain}
\newtheorem{theorem}{Theorem}
\newtheorem{lemma}[theorem]{Lemma}
\newtheorem{proposition}[theorem]{Proposition}

\theoremstyle{definition}
\newtheorem{definition}{Definition}

\theoremstyle{remark}
\newtheorem{remark}{Remark}
\newenvironment{IEEEproof}[1][Proof]{\begin{proof}[#1]}{\end{proof}}

\newcommand{\F}{\mathbb{F}}
\newcommand{\bbF}{\mathbb{F}}

\newcommand{\rankop}{\operatorname{rank}}
\newcommand{\kr}{\operatorname{kr}}

\newcommand{\rowsp}{\operatorname{rowsp}}

\newcommand{\Span}{\operatorname{span}}

\newcommand{\PGL}{\operatorname{PGL}}

\newcommand{\dstar}{d_{\min}^{\star}}

\providecommand{\Mob}{\operatorname{Mob}}
\providecommand{\id}{\operatorname{id}}

\begin{document}

\title{On Codes with Support-Constrained Parity Checks}

\author{
  Barron Han\\
  Computing and Mathematical Sciences\\
  California Institute of Technology\\
  Pasadena, CA, USA\\
  \texttt{bshan@caltech.edu}
  \and
  Hikmet Yildiz\\
  Independent Researcher\\
  Pasadena, CA, USA\\
  \texttt{hhikmetyildiz@gmail.com}
  \and
  Babak Hassibi\\
  Electrical Engineering\\
  California Institute of Technology\\
  Pasadena, CA, USA\\
  \texttt{hassibi@caltech.edu}
}
\date{}

\maketitle

\begin{abstract}
We study linear codes that maximize minimum distance subject to arbitrary support constraints on the parity-check matrix.
Such constraints arise naturally in the design of LDPC codes, locally repairable codes, and hardware-constrained systems where each parity check must involve only a limited number of code symbols. They are also essential in quantum error correction, where sparse stabilizers reduce measurement noise and respect the connectivity constraints of physical qubit architectures.
We derive the optimal minimum distance possible given support constraints on the parity-check matrix and show it is achievable over sufficiently large fields.
When this maximum distance coincides with the Singleton bound for unconstrained parity check matrices, the dual GM-MDS construction yields generalized Reed--Solomon codes obeying the mask.
In the generator-matrix setting, the GM-MDS theorem guarantees that the optimal distance can always be achieved by a subcode of a generalized Reed--Solomon code while satisfying arbitrary support constraints. We show that this is not true for the parity-check setting.
We exhibit a set of support constraints, derived from the vertex-edge incidence of $K_{6,6}$, for which the optimal minimum distance cannot be realized by any subcode of a generalized Reed--Solomon code over any field.
We also analyze structured constraint families---regular, balanced, and cyclic masks---through numerical optimization, providing design guidance for practical code constructions.

\end{abstract}

\section{Introduction}
\label{sec:intro}
The design of linear codes with large minimum distance under sparsity constraints is a well-studied problem. Support constraints on the generator matrix are particularly relevant for distributed storage systems and multiple access networks~\cite{dauMA, halbawidistribute}, as they restrict code symbols to linear combinations of message symbols from accessible nodes only. This leads to a fundamental question: given a prescribed set of support constraints, what is the maximum achievable minimum distance? Furthermore, can this optimum be achieved using structured codes over smaller fields? The GM-MDS theorem, conjectured in~\cite{DauGMMDS} and proven in~\cite{YildizGMMDS, LovettGMMDS}, addresses this by establishing that for any support constraint on the generator matrix of a length~$n$ code, the optimal minimum distance,~$d$, is achievable by a subcode of a generalized Reed--Solomon code with field size $q \geq 2n - d$.

This paper studies the dual problem of optimal minimum distance codes under sparsity constraints on the parity check matrix, which effectively limit each parity check equation to a linear combination of a prescribed set of code symbols. This framework is driven by two primary motivations. First, it models scenarios where parity checks must ``sparsely'' measure code symbols, such as limiting the frequency with which a specific symbol is accessed (column weight) or restricting the maximum number of terms per check equation to $w$ (row weight). Second, arbitrary sparsity patterns can be used to enforce topological design constraints. For instance, in physical circuit implementations, restricting bit interactions to local ``neighbors'' eliminates the need for additional long-range connections, thereby reducing hardware complexity.

In the classical setting, parity check sparsity is foundational to the theory of Low-Density Parity-Check (LDPC) codes~\cite{GallagerLDPC}. The sparse structure of the parity check matrix enables iterative message-passing decoders whose block complexity is linear in the number of edges per iteration, in contrast to the exponential complexity of maximum-likelihood decoding for general linear codes. The capacity-approaching performance of irregular LDPC codes~\cite{LDPCCap}, combined with their hardware-friendly decoding, has made them the channel codes of choice in modern communication standards: the 5G New Radio specification~\cite{3GPP5G} mandates quasi-cyclic LDPC codes with structured row and column weight constraints that are dictated by decoder architecture. More broadly, the study of codes with parity check sparsity constraints---including expander codes~\cite{SipserSpielman} and spatially coupled LDPC codes---reveals a fundamental tension between sparsity, minimum distance, and decoding complexity. Sparse parity checks reduce hardware cost and decoding latency, but a row weight of $w$ limits the minimum distance of any single check equation, and careful global design is needed to ensure that the code's minimum distance remains large.

A complementary motivation arises from distributed storage. In locally repairable codes (LRCs)~\cite{TamoBargLRC}, each code symbol must be recoverable from a small local group, which amounts to requiring low-weight parity checks supported within prescribed locality sets. Partial MDS codes~\cite{BlaumPMDS} further partition checks into local and global groups, imposing a coupled block diagonal sparsity pattern on the parity check matrix. In each of these settings, the sparsity pattern is dictated by the system architecture and cannot be freely chosen. The question we ask---\emph{what is the optimal minimum distance achievable under a given parity check support constraint?}---therefore subsumes these families as special cases.

These motivations extend naturally to quantum error correction. In the Calderbank-Shor-Steane (CSS) construction~\cite{CalderbankCSS,SteaneCSS}, the nonzero entries of the classical parity check matrices become stabilizer measurements on qubits. Dense parity check matrices therefore require high-weight stabilizers that couple many qubits, exacerbating noise and violating the connectivity constraints of physical architectures~\cite{RuizQuantum, SaboQuantum}. Sparse parity check matrices are thus essential for practical fault-tolerant quantum computing, as they enable low-weight, local stabilizer measurements.

While quantum analogs of LDPC~\cite{gottesmanLDPC}, Reed--Solomon~\cite{GrasslRS}, and Tamo--Barg~\cite{GolowichLRC, TamoLRC} codes can achieve linearly scaling distance and dimension, implementing them often requires long-range connections or dense parity checks. In practice, topological codes such as surface codes~\cite{KitaevSurface, FowlerSurface} use only local connections but achieve distance scaling only as $\sqrt{n}$. Designing codes with near-optimal distance, efficient decoding, and hardware-compatible sparse parity checks remains an important open problem towards fault tolerant quantum computation.

\subsection{Outline}

In Section~\ref{sec:expansion}, we formalize the sparsity-constrained parity-check setting, derive the optimal minimum distance from the support structure of the mask, and show that this distance is achievable over sufficiently large fields.
Section~\ref{sec:mds} treats the MDS regime: when the derived bound coincides with the Singleton bound, the dual GM-MDS construction yields a generalized Reed--Solomon code obeying the mask.
Section~\ref{sec:certificates} introduces a factorization of a Vandermonde matrix that certifies when a sparse parity-check code can be realized as a subcode of a generalized Reed--Solomon code.
Section~\ref{sec:counterexample} demonstrates the limits of this approach: we exhibit sparsity constraints arising from the vertex-edge incidence of the complete bipartite graph $K_{6,6}$ for which no subcode of a generalized Reed--Solomon code achieves the optimal distance over any field.
Finally, Section~\ref{sec:optimizing} studies structured constraint families---regular, balanced, and cyclic---and presents numerical optimization results.

\subsection{Notation}

For $n \geq 0$, $[n] := \{1, \ldots, n\}$.  Given sets $S_1, \ldots, S_n$ and $R \subset [n]$, we write $S(R) := \bigcup_{j\in R} S_j$.  An $[n,k,d]_q$ code is a linear code over $\F_q$ with length~$n$, dimension~$k$, and minimum distance~$d$.

\section{Sparse Parity-Check Masks and Exact Dimension}
\label{sec:expansion}
Let $\mathbb{F}_q$ be a finite field.  We use $m$ parity-check equations; the intended dimension is $n-m$, achieved when $\rankop H = m$.  A linear code $C \subseteq \mathbb{F}_q^n$ can be specified as the kernel of a parity-check matrix $H \in \bbF_q^{m \times n}$.  We adopt the convention $d_{\min}(\{0\}) = n+1$.  We impose sparsity via a binary mask
\begin{equation}
  M \in \{0,1\}^{m \times n},\qquad
  \begin{cases}
    M_{ij}=0 \implies H_{ij}=0,\\
    M_{ij}=1 \implies H_{ij}\ \text{is free}.
  \end{cases}
\end{equation}

For each column $j \in [n]$, define $S_j := \{i \in [m] : M_{ij} = 1\}$, and for $R \subseteq [n]$, define $S(R) := \bigcup_{j\in R} S_j$.

\begin{definition}[Structural row rank]
\begin{equation}
    \rho(M) := \max\{\rankop(H) : H_{ij}=0 \text{ whenever } M_{ij}=0\}.
\end{equation}
\end{definition}

Let $G_M$ be the bipartite graph with left vertices the rows, right vertices the columns, and edges at the free positions of $M$.  By Hall's theorem \cite{hallmarriage}, $\rho(M) = \nu(G_M)$, the maximum matching size in $G_M$.  In particular, $\rho(M) = m$ if and only if $|N(U)| \ge |U|$ for every $U \subseteq [m]$, where

\begin{equation}
  N(U) := \{j \in [n] : \exists\, i \in U,\, M_{ij}=1\}.
\end{equation}

\begin{definition}[Structural Kruskal rank]
    \begin{equation}
    \label{dmin_opt}
    \begin{split}
    \tau(M) = \max \{ & s \in \{0,1,\ldots,n\} : \ \forall R \subseteq [n], \\
    & |R| \leq s \implies |S(R)| \geq |R| \},
    \end{split}
    \end{equation}
\end{definition}

The case $\rho(M)=m$ is the full-structural-row-rank regime: the mask can support $m$ independent parity-check equations, so a filling with rank $m$ defines a code of the intended dimension $n-m$.

While the structural row rank determines whether the mask can support a code of the intended dimension $n - m$, the structural Kruskal rank controls the best achievable minimum distance.  Together, these quantities characterize the achievable parameters:

\begin{theorem}[Minimum distance of sparse parity check]
\label{thm:dminopt}
Let $M \in \{0,1\}^{m \times n}$ with column supports $S_j$.  Define
\[
  \begin{aligned}
  \dstar(M) &:= \tau(M)+1,\\
  \Delta(M) &:= \rho(M)+\tau(M)\binom{n}{\tau(M)}.
  \end{aligned}
\]

Then every filling $H$ obeying $M$ over any field satisfies $\kr(H) \le \tau(M)$, hence $d_{\min}(\ker H) \le \dstar(M)$.  Moreover, for every prime power $q>\Delta(M)$, there exists a filling $H\in\bbF_q^{m\times n}$ obeying $M$ with $\rankop(H) = \rho(M)$ and $d_{\min}(\ker H) = \dstar(M)$.

In particular, if $\rho(M) = m < n$, then $\ker H$ is an $[n, n{-}m, \dstar(M)]_q$ code obeying $M$.
\end{theorem}

\begin{IEEEproof}
Write $\tau=\tau(M)$ and $\rho=\rho(M)$.  We use the identity $d_{\min}(\ker H)=1+\kr H$, where $\kr H$ is the largest $r$ such that every set of at most $r$ columns of $H$ is linearly independent.

\emph{Upper bound.}  If $\tau<n$, maximality of $\tau$ in \eqref{dmin_opt} gives $R_0\subseteq[n]$ with $|R_0|=\tau+1$ and $|S(R_0)|<|R_0|$.  For any $H$ obeying $M$, the columns $H_{:,R_0}$ are supported on $S(R_0)$, so $\rankop(H_{:,R_0})\le |S(R_0)|<|R_0|$ and thus $\kr H\le\tau$.  If $\tau=n$, the bound $d_{\min}(\ker H) \le n+1$ is immediate from the convention $d_{\min}(\{0\})=n+1$.

\emph{Achievability.}  The cases $\tau=0$ and $\tau=n$ are immediate.  For $0<\tau<n$, let $X$ be the symbolic matrix with an independent variable $x_{ij}$ at each allowed position.

For each $R\subseteq[n]$ with $|R|=\tau$, the expansion hypothesis gives $|S(J)|\ge|J|$ for all $J\subseteq R$, so Hall's theorem yields a matching of $R$ into $S(R)$.  Let $I_R$ be the matched rows and set $p_R:=\det X_{I_R,R}$.  The matching contributes a monomial $\prod_{j\in R}x_{\mu(j),j}$ that no other permutation can produce (the variables are indexed by distinct positions), so $p_R$ is nonzero with $\deg p_R=\tau$.

Similarly, since $\rho(M) = \nu(G_M)$, choose a maximum matching of size $\rho$ in $G_M$.  The corresponding minor $p_0:=\det X_{I_0,J_0}$ is nonzero with $\deg p_0=\rho$.

Form $P:=p_0\prod_{|R|=\tau}p_R$.  This is nonzero with $\deg P=\rho+\tau\binom{n}{\tau}=\Delta(M)$.  Evaluate the free entries of $X$ independently and uniformly over $\mathbb{F}_q$.  By the Schwartz--Zippel lemma,
\[
  \Pr[P(H)=0] \le \frac{\Delta(M)}{q} < 1
  \qquad \text{for } q>\Delta(M).
\]
Thus some evaluation $H$ has $P(H)\ne 0$.

Since $p_0(H)\ne0$, we get $\rankop(H)\ge\rho$; by definition of $\rho(M)$, equality holds.  Since $p_R(H)\ne0$ for every $\tau$-subset $R$, every such subset is independent; since any smaller subset extends to a $\tau$-subset and a subset of an independent set is independent, $\kr H\ge\tau$.  The upper bound gives $\kr H=\tau$.  Therefore $d_{\min}(\ker H)=\tau+1=\dstar(M)$ and $\dim\ker H=n-\rho(M)$, which equals $n-m$ when $\rho(M)=m$.
\end{IEEEproof}

Theorem~\ref{thm:dminopt} is a generic existence result: it identifies the best distance permitted by the support pattern alone but does not provide an efficient decoder or a small-field algebraic construction.  The remaining question is whether this generic optimum can be realized as a subcode of a structured parent family---especially a generalized Reed--Solomon code---over smaller fields.

\section{The MDS Regime and the Dual GM-MDS Construction}
\label{sec:mds}
For the existence of an MDS code, a necessary condition is that the expansion condition of Theorem~\ref{thm:dminopt} holds through size $m$:
\begin{equation}
\label{mds}
    |R| \le m \implies |S(R)| \ge |R|
    \qquad\text{for every } R \subseteq [n].
\end{equation}

\begin{proposition}[Sparse parity-check MDS construction]
\label{prop:mds}
Let $M \in \{0,1\}^{m \times n}$ satisfy \eqref{mds}.  Then, for every prime power $q \geq n + m - 1$, there exists a generalized Reed--Solomon $[n, n{-}m]_q$ code whose parity-check matrix obeys $M$.
\end{proposition}

\begin{IEEEproof}
Regard $M$ as a generator support mask for an $[n,m]$ code: the column-expansion condition $|S(R)| \ge |R|$ for all $R \subseteq [n]$ with $|R| \le m$ is exactly the GM-MDS zero-pattern condition on the $m \times n$ generator matrix, since the allowed zero positions in each column of the generator are the complement of $S_j$.  The GM-MDS theorem~\cite{YildizGMMDS, LovettGMMDS} then yields a generalized Reed--Solomon $[n,m,n{-}m{+}1]_q$ code with generator obeying $M$ for $q \ge n+m-1$.  Its dual is again generalized Reed--Solomon, with parity-check matrix obeying $M$.
\end{IEEEproof}

When \eqref{mds} holds, the GM-MDS theorem provides a recipe: construct a generalized Reed--Solomon generator obeying $M$ and take its dual.  However, when $\dstar(M) < m+1$, the GM-MDS theorem does not directly yield generalized Reed--Solomon constructions satisfying the parity-check constraints.  A natural question is whether the sparse code can still be realized as a subcode of a generalized Reed--Solomon code; the next section formalizes this.

\section{Vandermonde Certificates}
\label{sec:certificates}
Theorem~\ref{thm:dminopt} guarantees that a filling of the mask achieving the optimal minimum distance \emph{exists} over sufficiently large fields, but the proof is non-constructive and relies on an exponentially large in $n$ field size.  In the generator-matrix setting, the GM-MDS theorem~\cite{YildizGMMDS, LovettGMMDS} shows that the optimal distance can be achieved by subcodes of generalized Reed--Solomon codes under support constraints. Subcodes of a well-understood parent family are desirable because they inherit both the minimum distance guarantee and the efficient decoding algorithms of the parent code. This motivates the following question:
\begin{center}
\emph{Can a sparse parity-check code always be realized as a subcode of a generalized Reed--Solomon code?}
\end{center}

\begin{definition}[Full-rank Vandermonde certificate]
\label{def:vand-cert}
Fix a target distance $d$ and set $r = d-1$.  A mask $M$ admits a \emph{full-rank generalized Vandermonde certificate of distance $d$} over $\bbF_q$ if there exist matrices $H \in \bbF_q^{m \times n}$ and $A \in \bbF_q^{r \times m}$ such that $H$ obeys $M$, $\rankop(H) = m$, $\rankop(A) = r$, and $AH = V$ for some Vandermonde matrix $V \in \bbF_q^{r \times n}$ with distinct evaluation points.
\end{definition}

\begin{remark}
    A generalized Reed--Solomon parity-check matrix has nonzero column multipliers; these can be absorbed into~$H$ by rescaling columns without changing the support pattern, so $V$ may be taken to be a plain Vandermonde matrix without loss of generality.
\end{remark}

\begin{lemma}[Certificate equivalence]
\label{lem:subcode}
Let $H \in \bbF^{m \times n}$ have full row rank and let $V \in \bbF^{r \times n}$.  Then $\ker(H) \subseteq \ker(V)$ if and only if there exists $A \in \bbF^{r \times m}$ with $V = AH$.  If $\rankop(V) = r$, then $\rankop(A) = r$.
\end{lemma}

\begin{IEEEproof}
If $V = AH$ and $Hc = 0$, then $Vc = A(Hc) = 0$.  Conversely, $\ker(H) \subseteq \ker(V)$ implies $\rowsp(V) \subseteq \rowsp(H)$ by nullspace--rowspace duality, so $V = AH$ for some $A$.
\end{IEEEproof}

\section{A Full-Rank Mask with No Vandermonde Certificate}
\label{sec:counterexample}
We answer the question posed in Section~\ref{sec:certificates} in the negative.
We exhibit a mask that admits no generalized Vandermonde certificate (Definition~\ref{def:vand-cert}) of the optimal distance given by Theorem \ref{thm:dminopt} over any field.

Let $m=12$, $n=36$, $r=5$, $d=6$.  Index the rows by $L_1, \ldots, L_6, R_1, \ldots, R_6$ and the columns by ordered pairs $c_{ij}$ for $1 \le i,j \le 6$.
Define column $c_{ij}$ to be allowed exactly in rows $L_i$ and $R_j$:
\[
  S_{ij} = \{L_i, R_j\}.
\]
Equivalently, $M$ is the vertex-edge incidence matrix of $K_{6,6}$, the complete bipartite graph with $6$ left and right nodes.
With columns grouped by left vertex,
\[
M_{K_{6,6}} =
\left[
\begin{array}{cccc}
\mathbf{1}_6^\top & 0 & \cdots & 0\\
0 & \mathbf{1}_6^\top & \cdots & 0\\
\vdots & \vdots & \ddots & \vdots\\
0 & 0 & \cdots & \mathbf{1}_6^\top\\ \hline
I_6 & I_6 & \cdots & I_6
\end{array}
\right],
\]
where $\mathbf{1}_6$ is the all-one column vector and $I_6$ is the $6 \times 6$ identity matrix.
The first six rows are $L_1, \ldots, L_6$, and the last six rows are $R_1, \ldots, R_6$.

\begin{proposition}[Column expansion]
\label{prop:k66-expansion}
$\dstar(M_{K_{6,6}}) = 6$.
\end{proposition}

\begin{IEEEproof}
Any edge set of size $\le 5$ in $K_{6,6}$ spans at least as many vertices: a bipartite graph on $v\le4$ vertices has at most $\lfloor v^2/4\rfloor\le v$ edges.
The six edges $\{c_{ij}: i \in \{1,2\}, j \in \{1,2,3\}\}$ span only five vertices.
\end{IEEEproof}

\begin{proposition}[Full structural row rank]
\label{prop:k66-rho}
$\rho(M_{K_{6,6}}) = 12$.
\end{proposition}

\begin{IEEEproof}
Match $L_i \mapsto c_{ii}$ and $R_j \mapsto c_{j+1,j}$ (indices mod~$6$).
These twelve columns are distinct and each is incident to its matched row.
\end{IEEEproof}

Together, Propositions~\ref{prop:k66-expansion} and~\ref{prop:k66-rho} show that $\tau(M) = 5$ and $\rho(M) = m = 12$.
By Theorem~\ref{thm:dminopt}, there exists a $[36, 24, 6]_q$ code obeying $M_{K_{6,6}}$ over any sufficiently large field.
We now show that no such code can be realized as a subcode of a generalized Reed--Solomon code.

\begin{theorem}
\label{thm:no-cert}
$M_{K_{6,6}}$ admits no full-rank Vandermonde certificate of distance~$6$ over any field.
\end{theorem}

\begin{IEEEproof}[Proof sketch]
Assume for contradiction that a full-rank Vandermonde certificate exists. It suffices to rule out certificates over an algebraically closed field.  After absorbing the nonzero column multipliers into $H$, write
\[
  \nu(t)=(1,t,t^2,t^3,t^4)^T,
  \qquad
  \nu(t_{ij})=x_{ij}a_i+y_{ij}b_j,
\]
where the $36$ values $t_{ij}$ are distinct and $a_i,b_j\in\mathbb{F}^5$ are the columns of $A$ indexed by the left and right vertices of $K_{6,6}$. This would imply that the vectors $\nu(t_{ij})$ for distinct evaluation points lie in the span of $a_i$ and $b_j$.

Fix $i\ne k$ and project away from $\operatorname{span}(a_i,a_k)$ by choosing $W\in\mathbb{F}^{3\times 5}$ with kernel $\operatorname{span}(a_i,a_k)$ and set $q(t)=W\nu(t)$.  Then $q(t_{ij})$ and $q(t_{kj})$ are proportional for all $j$.  The key algebraic lemma, proved in Appendix~\ref{app:k66-proof}, says that these six proportionalities force the existence of fractional-linear involutions $\tau_{ik}\in\Mob(\mathbb{F})\cong\PGL_2(\mathbb{F})$ such that
\[
  t_{kj}=\tau_{ik}(t_{ij}),\qquad j=1,\ldots,6 .
\]

Analyzing this family of involutions demonstrates that it eventually forces too many of the vectors $\nu(t_{ij})$ to lie in a low-dimensional subspace. This violates the fact that any four such vectors with distinct evaluation points must remain linearly independent. The full algebraic details are given in Appendix~\ref{app:k66-proof}.
\end{IEEEproof}

The mask $M_{K_{6,6}}$ satisfies $\rho(M) = 12$ and $\dstar(M) = 6$.
Over large fields, generic methods produce an exact $[36, 24, 6]$ code obeying the mask.
Yet no such code can be realized as a subcode of a generalized Reed--Solomon code via a Vandermonde factorization.

\section{Structured Mask Optimization}
\label{sec:optimizing}
While the general condition of Theorem~\ref{thm:dminopt} is difficult to analyze, structured mask families admit sharper results.

\begin{definition}[Regular]
\label{def:regular}
A mask $M \in \{0,1\}^{m \times n}$ is $w$-\emph{regular} if each row has exactly $w$ ones.
\end{definition}

\begin{definition}[Balanced]
\label{def:balanced}
A $w$-regular mask is \emph{balanced} if each column has $\lfloor wm/n \rfloor$ or $\lceil wm/n \rceil$ ones.
\end{definition}

\begin{definition}[Cyclic]
\label{def:cyclic}
A mask is \emph{cyclic} if each row is a cyclic shift of the first row.
\end{definition}

A cyclic mask is always $w$-regular.  Regular masks simplify hardware since each parity check uses exactly $w$ terms.  Balanced masks distribute measurements evenly---important in quantum systems where repeated measurement introduces noise~\cite{RuizQuantum, SaboQuantum}.

\begin{proposition}
\label{prop:reg-mds}
For any $w \geq n-m+1$ and every prime power $q \ge n+m-1$, there exist generalized $[n,n{-}m]_q$ Reed--Solomon codes satisfying a $w$-regular, balanced, and cyclic mask.
\end{proposition}

\begin{IEEEproof}
The dual codes of the sparse Reed--Solomon constructions in~\cite{halbawisparse} satisfy these constraints by Proposition~\ref{prop:mds}.
\end{IEEEproof}

\subsection{Numerical Results}

Figures~\ref{fig:regular}--\ref{fig:cycbal} display the distance
\begin{equation}
    \begin{split}
    D_{\mathrm{fr}}(m,n,w) := \max\{ \dstar(M) :\;&
    M \in \mathcal{M}_{\mathrm{reg}}(w),\\
    & \rho(M)=m\},
    \end{split}
\end{equation}
computed by exhaustive search over all mask patterns for $n = 9$ under regular and cyclic constraints, and for $n = 16$ under balanced-cyclic constraints.  Each mask defines an $[n,n{-}m]$ code over sufficiently large fields by Theorem \ref{thm:dminopt}.

\begin{figure}[ht]
    \centering
    \includegraphics[width=0.52\linewidth]{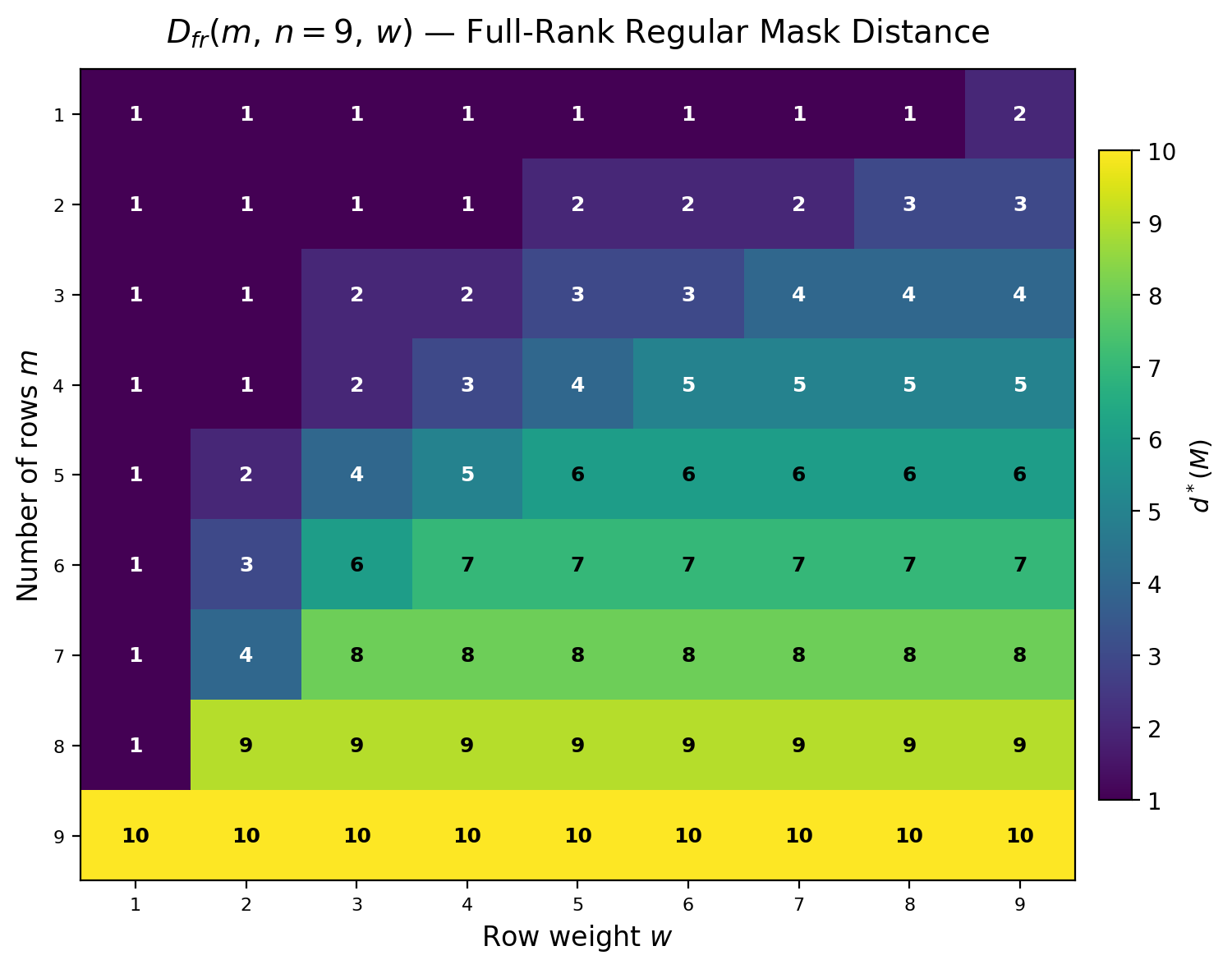}
    \caption{Full-rank $w$-regular mask distance $D_{\mathrm{fr}}(m,9,w)$.}
    \label{fig:regular}
\end{figure}

\begin{figure}[ht]
    \centering
    \includegraphics[width=0.52\linewidth]{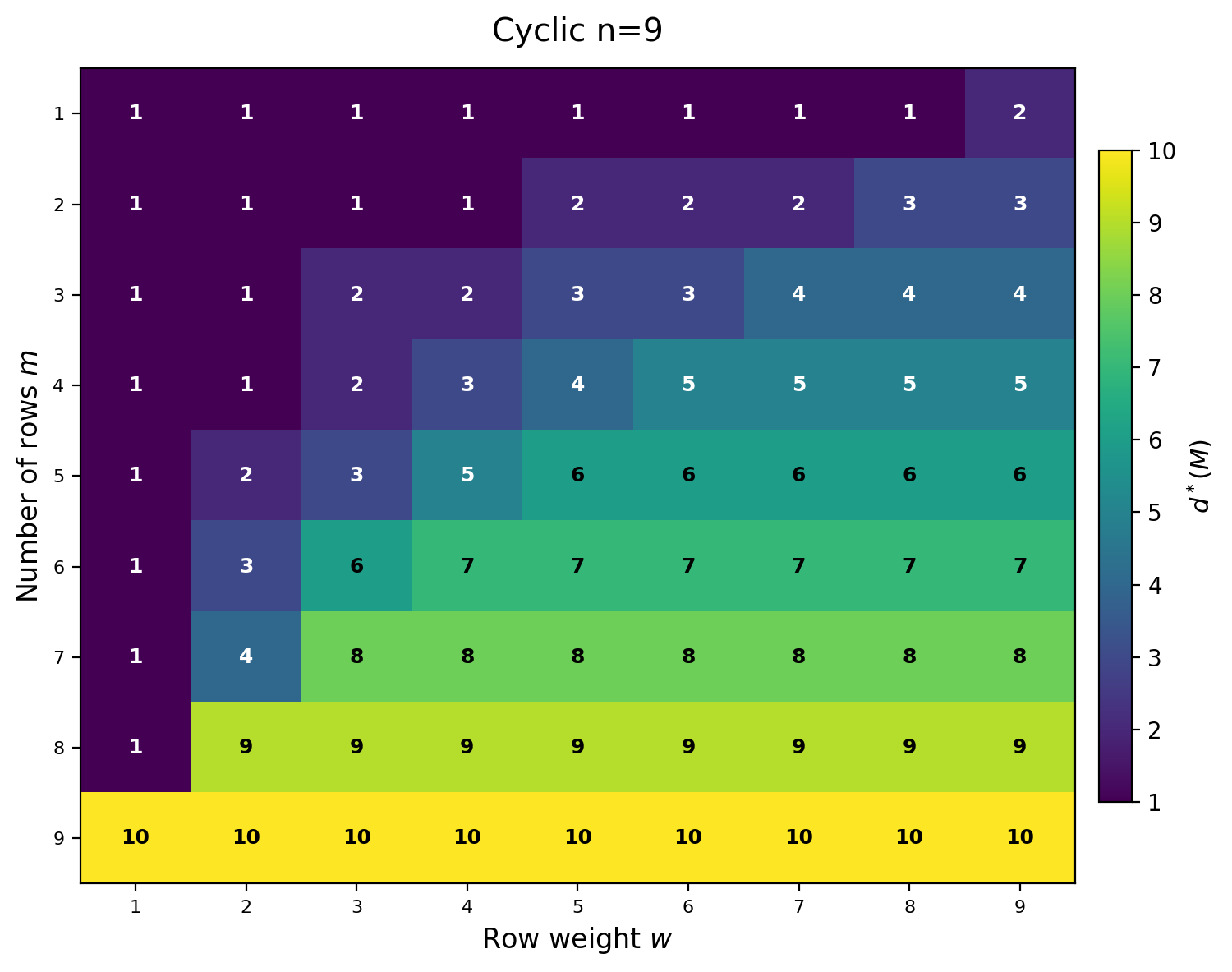}
    \caption{Full-rank cyclic mask distance for $n=9$.  Comparing with Figure~\ref{fig:regular}, the cyclic optimum at $(m,w)=(5,3)$ is only $3$ versus $4$.}
    \label{fig:cyclic}
\end{figure}

\begin{figure}[ht]
    \centering
    \includegraphics[width=0.7\linewidth]{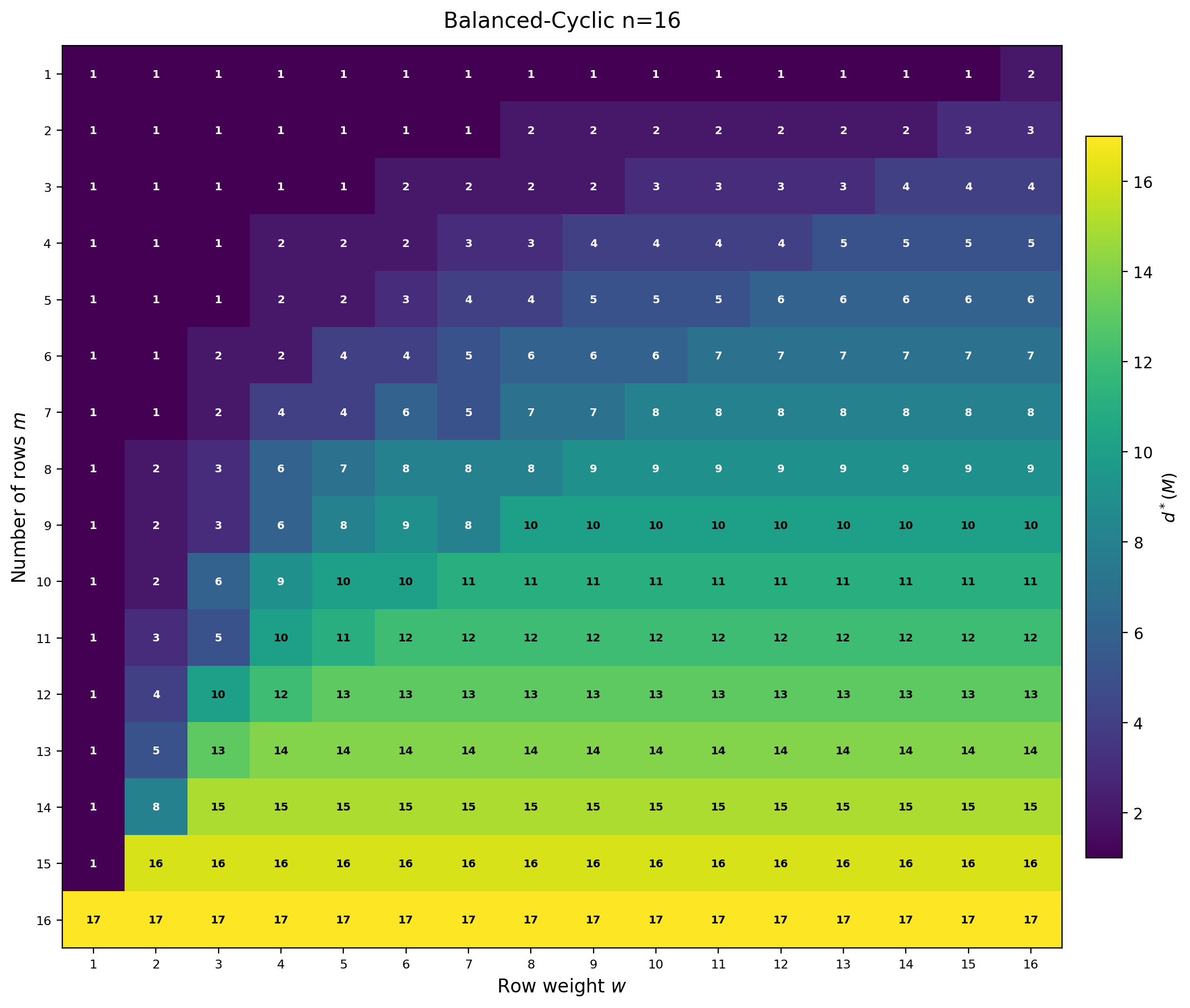}
    \caption{Full-rank balanced-cyclic mask distance for $n=16$.  The $m=9$ row exhibits non-monotonicity: $d_{\min}^{\star}(9,16,6)=9$ but $d_{\min}^{\star}(9,16,7)=8$.}
    \label{fig:cycbal}
\end{figure}

\subsection{Observations}

\begin{remark}[Insufficiency of cyclic masks]
\label{rem:cyclic}
Let $d_{\min}^{\star}(m,n,w)$ and $d_{\mathrm{cyclic}}^{\star}(m,n,w)$ denote the maximum expansion-bound distance over all $w$-regular masks and all cyclic masks, respectively.  There exist parameters for which $d_{\mathrm{cyclic}}^{\star}(m,n,w) < d_{\min}^{\star}(m,n,w)$.  For instance, Figures~\ref{fig:regular} and~\ref{fig:cyclic} show $d_{\mathrm{cyclic}}^{\star}(5,9,3) = 3$ while $d_{\min}^{\star}(5,9,3) = 4$.  Thus cyclic structure alone is not sufficient to maximize the expansion-bound distance.
\end{remark}

\begin{remark}[Non-monotonicity of balanced-cyclic masks]
\label{rem:nonmono}
The expansion-bound distance is not monotone in the row weight $w$.  Figure~\ref{fig:cycbal} exhibits this for balanced-cyclic masks at $n=16$: $d_{\min}^{\star}(9,16,6)=9$ but $d_{\min}^{\star}(9,16,7)=8$.
\end{remark}

\begin{definition}[Graph incidence expansion]
For a multigraph $G=(V,E)$ and $R\subseteq V$, let
\[
  E(R) := \{e\in E : e\cap R\ne\emptyset\}
\]
be the set of edges incident to at least one vertex of $R$, counted with multiplicity.  The incidence expansion of $G$, denoted $\eta(G)$, is
\begin{equation}
    \begin{split}
        \eta(G) = \max \left\{ \right. & s \in \mathbb{Z}_{\geq 0} : \ \forall R \subseteq V, \\
        & \left. 1 \le |R| \le s \implies |E(R)| \ge |R| \right\}.
    \end{split}
\end{equation}
\end{definition}

\begin{proposition}[$2$-regular masks are limited to repetition-code distance]
\label{prop:w2}
Let $m < n$.  A $2$-regular mask $M$ with $m$ rows and $n$ columns satisfies
\begin{equation}
  \dstar(M) \le \left\lfloor \frac{n}{n-m} \right\rfloor .
\end{equation}
Equality is achieved by a cyclic mask (Def. \ref{def:cyclic}).  In particular, $2$-regular parity-check constraints can never exceed the minimum distance of the $[n, n{-}m]$ block repetition code.
\end{proposition}

\begin{IEEEproof}
A $2$-regular mask corresponds to a multigraph $G=(V,E)$ with $|V|=n$ and $|E|=m$, where columns are vertices and rows are edges.  Under this correspondence, $\tau(M)=\eta(G)$.  Any such graph has at least $n{-}m$ acyclic connected components (each cycle closes at most one excess edge).  Let $V_1,\ldots,V_{n-m}$ be the vertex sets of these components.  By pigeonhole, some $V^*$ satisfies $|V^*|\le\lfloor n/(n{-}m)\rfloor$.  Since $V^*$ is an entire acyclic component, its incident edge set equals its internal edge set, so $|E(V^*)|=|V^*|-1<|V^*|$.  Hence $\eta(G)\le\lfloor n/(n{-}m)\rfloor-1$ and $\dstar(M)\le\lfloor n/(n{-}m)\rfloor$.

For tightness, delete $n{-}m$ nearly evenly spaced edges from the $n$-cycle, leaving $n{-}m$ disjoint paths with balanced numbers of vertices.  For any $R$ with $|R|\le\lfloor n/(n{-}m)\rfloor-1$, partition $R_i=R\cap P_i$; since $R_i\subsetneq P_i$, each piece satisfies $|E(R_i)|\ge|R_i|$, so $|E(R)|\ge|R|$.  Thus $\dstar(M)=\lfloor n/(n{-}m)\rfloor$.  This mask is cyclic and has $\rho(M)=m$ (every forest has a matching saturating all edges).
\end{IEEEproof}

\begin{remark}[Conditions for $d_{\min}^{\star} \leq 2$]
\label{rem:dstar-small}
For $w$-regular masks: (1) $d_{\min}^{\star}(m,n,w) = 1$ if and only if $m < n/w$, since the $mw$ ones cannot cover all $n$ columns; (2) $d_{\min}^{\star}(m,n,w) = 2$ if and only if $n/w \leq m < 2n/(w{+}1)$, since a pigeonhole argument forces two weight-$1$ columns to share a row.
\end{remark}

\clearpage

\section{Conclusion}
\label{sec:conclusion}
We derived a tight expansion-based bound on the minimum distance of sparse parity-check codes and connected the MDS regime to the dual GM-MDS construction.  We introduced generalized Vandermonde certificates for realizing sparse codes as subcodes of generalized Reed--Solomon codes and showed via the $K_{6,6}$ counterexample that, unlike the generator-matrix setting, such certificates do not always exist.

Several directions remain open.  Determining tight necessary and sufficient conditions on the mask for the existence of a Vandermonde certificate is a natural next step; we have obtained explicit generalized Reed--Solomon subcode constructions under certain mask restrictions, which will appear in subsequent work.  The certificate framework itself generalizes to other structured code families.  Finally, the CSS construction~\cite{CalderbankCSS,SteaneCSS} converts pairs of dual-contained classical codes into quantum codes; extending the sparse parity-check framework to this setting is a promising direction for quantum codes with structured stabilizers.

\clearpage

\appendix
\section{Proof of Theorem~\ref{thm:no-cert}}
\label{app:k66-proof}
\subsection{Algebraic reduction}

We prove Theorem~\ref{thm:no-cert}.  It suffices to rule out certificates over an
algebraically closed field $\F$: if a certificate exists over a field $F$, then after
extending scalars to $\overline F$ the same equations hold, ranks are preserved, and the
Vandermonde evaluation points remain distinct.

Assume, for contradiction, that such a certificate exists.  Write the columns of $A$
corresponding to rows $L_i$ as $a_i$, and those corresponding to rows $R_j$ as $b_j$.
After absorbing the nonzero column multipliers into $H$, the Vandermonde columns have the
form
\[
\nu(t)=(1,t,t^2,t^3,t^4)^T.
\]
Since column $(i,j)$ of $H$ is supported only on rows $L_i$ and $R_j$, there are scalars
$x_{ij},y_{ij}\in \F$, not both zero, such that
\begin{equation}
\label{eq:basic-certificate}
\nu(t_{ij})=x_{ij}a_i+y_{ij}b_j,
\qquad 1\leq i,j\leq 6,
\end{equation}
where the $36$ evaluation points $t_{ij}$ are pairwise distinct.

\subsection{Algebraic preliminaries}

\begin{lemma}[Moment-curve independence]
\label{lem:moment-independence}
If $t_1,\ldots,t_s\in \F$ are distinct and $s\leq 5$, then
\[
\nu(t_1),\ldots,\nu(t_s)\in \F^5
\]
are linearly independent.  Consequently:
\begin{enumerate}[label=(\alph*)]
\item two distinct vectors of the form $\nu(t)$ are not proportional;
\item if $p_1,p_2,p_3,p_4$ are four distinct vectors of the form $\nu(t)$, then
\[
\Span(p_1,p_2)\cap \Span(p_3,p_4)=\{0\}.
\]
\end{enumerate}
\end{lemma}

\begin{IEEEproof}
The $s\times s$ submatrix obtained from the first $s$ rows is a Vandermonde matrix with
nonzero determinant
\[
\prod_{1\leq a<b\leq s}(t_b-t_a).
\]
Thus the $s$ columns are linearly independent.  The two consequences follow immediately.
\end{IEEEproof}

\begin{lemma}[Nonzero and pairwise independent left vectors]
\label{lem:left-independent}
In any hypothetical certificate, the six left-row vectors $a_1,\ldots,a_6\in \F^5$ are
nonzero and pairwise linearly independent.
\end{lemma}

\begin{IEEEproof}
Use the notation of \eqref{eq:basic-certificate}.
First suppose $a_i=0$ for some $i$.  Then each $\nu(t_{ij})$ is a nonzero scalar multiple of
$b_j$.  If also $a_k=0$ for some $k\neq i$, then $\nu(t_{ij})$ and $\nu(t_{kj})$ are
proportional for every $j$, contradicting Lemma~\ref{lem:moment-independence}.  Hence
choose $k\neq i$ with $a_k\neq 0$.

For each $j$, the two independent vectors $\nu(t_{ij})$ and $\nu(t_{kj})$ lie in the same
two-dimensional space $\Span(a_k,b_j)$; if $\Span(a_k,b_j)$ were one-dimensional, they
would be proportional.  Thus they form a basis of $\Span(a_k,b_j)$, and in particular
\[
a_k\in \Span(\nu(t_{ij}),\nu(t_{kj}))
\]
for all six values of $j$.  Taking two different values of $j$ gives two spans of four
distinct vectors of the form $\nu(t)$ meeting in the nonzero vector $a_k$, contradicting
Lemma~\ref{lem:moment-independence}.  Therefore every $a_i$ is nonzero.

Now suppose $a_i$ and $a_k$ are proportional for some $i\neq k$, and let $a$ be a common
nonzero vector spanning them.  For each $j$, the two independent vectors
$\nu(t_{ij})$ and $\nu(t_{kj})$ lie in $\Span(a,b_j)$.  As above, they form a basis of this
space, so
\[
a\in \Span(\nu(t_{ij}),\nu(t_{kj}))
\]
for all $j$.  Again two different values of $j$ contradict Lemma~\ref{lem:moment-independence}.
Hence $a_i$ and $a_k$ are linearly independent.
\end{IEEEproof}

\begin{lemma}[Algebraic normal form for irreducible ternary quadratics]
\label{lem:quadratic-normal-form}
Let $Q(X,Y,Z)\in \F[X,Y,Z]$ be an irreducible homogeneous polynomial of degree $2$.
Then there are independent linear forms $L_1,L_2,L_3$ in $X,Y,Z$ such that
\[
Q=L_2^2-L_1L_3.
\]
In characteristic two, this is the same as $Q=L_2^2+L_1L_3$.
\end{lemma}

\begin{IEEEproof}
First choose a nonzero vector $p\in \F^3$ with $Q(p)=0$.  This is elementary: restrict
$Q$ to any two-dimensional subspace.  If the restriction is the zero polynomial, any nonzero
vector in that subspace works.  Otherwise the restriction is a nonzero binary homogeneous
quadratic, which factors into linear forms over the algebraically closed field $\F$ and hence
has a nonzero zero.

Extend $p$ to a basis $p,u,v$ of $\F^3$.  In the corresponding coordinates, replace $Q$ by
\[
\widetilde Q(X,Y,Z):=Q(Xp+Yu+Zv).
\]
Since $\widetilde Q(1,0,0)=Q(p)=0$, there is no $X^2$ term.  Thus
\[
\widetilde Q=X\ell(Y,Z)+q(Y,Z),
\]
where $\ell$ is linear and $q$ is homogeneous quadratic in $Y,Z$.

If $\ell=0$, then $\widetilde Q=q(Y,Z)$ is a binary homogeneous quadratic, so it factors into
linear forms over $\F$, contradicting irreducibility.  Hence $\ell\neq 0$.  After an invertible
linear change in the variables $Y,Z$, assume $\ell=Z$.  Then
\[
\widetilde Q=XZ+aY^2+bYZ+cZ^2.
\]
If $a=0$, then
\[
\widetilde Q=Z(X+bY+cZ),
\]
again reducible.  Hence $a\neq 0$.  Since $\F$ is algebraically closed, rescale $Y$ so that
$a=1$.  Replacing $X$ by $X-bY-cZ$ gives
\[
\widetilde Q=Y^2+XZ.
\]
Equivalently, after replacing $X$ by $-X$ when the characteristic is not two,
\[
\widetilde Q=Y^2-XZ.
\]
Undoing the change of variables gives $Q=L_2^2-L_1L_3$ for independent linear forms
$L_1,L_2,L_3$.
\end{IEEEproof}

\begin{lemma}[Polynomial triples satisfying $F_2^2=F_1F_3$]
\label{lem:square-factorization}
Let $F_1,F_2,F_3\in \F[t]$ be linearly independent polynomials with
\[
\gcd(F_1,F_2,F_3)=1,
\qquad
F_2^2=F_1F_3.
\]
Then there are coprime polynomials $R,S\in \F[t]$ and nonzero scalars
$\lambda,\mu,\eta\in \F^\times$ such that
\[
F_1=\lambda R^2,
\qquad
F_2=\eta RS,
\qquad
F_3=\mu S^2,
\qquad
\eta^2=\lambda\mu.
\]
If $M_0=\max_i\deg F_i$ and $h=\max(\deg R,
\deg S)$, then
\[
2h\leq M_0,
\qquad
h\neq 0.
\]
\end{lemma}

\begin{IEEEproof}
The identity $F_2^2=F_1F_3$ and the condition $\gcd(F_1,F_2,F_3)=1$ imply
$\gcd(F_1,F_3)=1$.  Indeed, any common divisor of $F_1$ and $F_3$ divides $F_2^2$, hence
also divides $F_2$, contradicting the gcd condition.

Because $F_1F_3$ is a square and $F_1,F_3$ are coprime, unique factorization in $\F[t]$
gives
\[
F_1=\lambda R^2,
\qquad
F_3=\mu S^2
\]
for coprime $R,S$ and nonzero $\lambda,\mu\in \F^\times$.  Then
\[
F_2^2=\lambda\mu R^2S^2.
\]
Since $\F$ is algebraically closed, choose $\eta\in\F^\times$ with $\eta^2=\lambda\mu$;
then $F_2=\pm\eta RS$, and the sign may be absorbed into $\eta$.

The degree inequality follows because $F_1$ and $F_3$ have degrees $2\deg R$ and
$2\deg S$.  If $h=0$, then $R,S$ are constants, so all three $F_i$ are constant multiples of
one another, contradicting their linear independence.  Hence $h\neq0$.
\end{IEEEproof}

For the next lemma, write
\[
\widehat{\F}:=\F\cup\{\infty\}.
\]
A fractional linear transformation is a function on $\widehat\F$ of the form
\[
t\longmapsto \frac{\alpha t+\beta}{\gamma t+\delta},
\qquad
\alpha\delta-\beta\gamma\neq0,
\]
with the usual conventions at the pole and at $\infty$.  Let $\Mob(\F)$ denote this group;
equivalently, it is the action of $\PGL_2(\F)$ on $\widehat{\F}$.

\begin{lemma}[Quadratic quotients have a residual fractional-linear involution]
\label{lem:degree-two-involution}
Let $R,S\in \F[t]$ be coprime with
\[
\max(\deg R,\deg S)=2.
\]
Set
\[
N(T,U):=R(T)S(U)-R(U)S(T).
\]
Assume there exist distinct $t,u\in\F$ with $N(t,u)=0$.  Then there is a nonidentity
involution $\tau\in\Mob(\F)$ such that, whenever $p\neq q$ and $N(p,q)=0$, one has
$q=\tau(p)$.
\end{lemma}

\begin{IEEEproof}
Since $N(T,T)=0$, the polynomial $T-U$ divides $N(T,U)$.  Write
\[
N(T,U)=(T-U)B(T,U).
\]
Because $R$ and $S$ have degree at most $2$, the quotient $B$ has degree at most $1$ in
each variable.  Also $N(U,T)=-N(T,U)$ and $U-T=-(T-U)$, so $B(U,T)=B(T,U)$.  Hence
\[
B(T,U)=aTU+b(T+U)+d
\]
for some $a,b,d\in\F$.

We claim $b^2-ad\neq0$.  If $b^2-ad=0$, then the symmetric bilinear polynomial $B(T,U)$
has rank at most one.  If $B$ is constant, then $N(T,U)$ is a scalar multiple of $T-U$,
which forces the rational function $R/S$ to have degree at most one; this contradicts
$\max(\deg R,\deg S)=2$ with $R,S$ coprime.  Otherwise, over the algebraically closed field
$\F$, the polynomial $B(T,U)$ is a nonzero scalar multiple of $\ell(T)\ell(U)$ for some
nonconstant linear polynomial $\ell$.  Choose $t_0$ with $\ell(t_0)=0$.  Then $B(t_0,U)=0$
for every $U$, so $N(t_0,U)=0$ for every $U$.  Thus
\[
R(t_0)S(U)-R(U)S(t_0)=0
\]
for every $U$.  Since $R$ and $S$ have no common root, $(R(t_0),S(t_0))\neq(0,0)$; the last
identity forces $R$ and $S$ to be scalar multiples, contrary to
$\max(\deg R,\deg S)=2$ and coprimality.  Therefore $b^2-ad\neq0$.

Define
\[
\tau(t):=-\frac{bt+d}{at+b}.
\]
The determinant condition $b^2-ad\neq0$ says that this is a fractional linear transformation.
For $p\neq q$, the equation $N(p,q)=0$ is equivalent to $B(p,q)=0$, and
\[
\begin{aligned}
B(p,q)=0
&\quad\Longleftrightarrow\quad
(ap+b)q+(bp+d)=0\\
&\quad\Longleftrightarrow\quad
q=\tau(p),
\end{aligned}
\]
with the usual interpretation when $ap+b=0$.

Finally $B(T,U)=B(U,T)$ implies
\[
q=\tau(p)\quad\Longleftrightarrow\quad p=\tau(q),
\]
so $\tau^2=\id$.  The assumed existence of a distinct pair $t\neq u$ with $N(t,u)=0$ implies
$u=\tau(t)\neq t$, so $\tau\neq\id$.
\end{IEEEproof}

\begin{lemma}[Characteristic-two translation invariant quadratics]
\label{lem:translation-invariant-quadratics}
Assume $\operatorname{char}\F=2$ and $a\in\F^\times$.  Let $R,S\in\F[t]$ be coprime
polynomials of degree at most $2$ satisfying
\[
R(t+a)S(t)-R(t)S(t+a)\equiv0.
\]
Then $R$ and $S$ both lie in the two-dimensional space
\[
\Span\{1,\ t^2+at\}.
\]
If $R,S$ are linearly independent, then
\[
\Span\{R^2,RS,S^2\}=\Span\{1,
\ t^2+at,
\ (t^2+at)^2\}.
\]
\end{lemma}

\begin{IEEEproof}
The displayed identity gives
\[
R(t+a)S(t)=R(t)S(t+a).
\]
Since $\gcd(R(t),S(t))=1$, the polynomial $R(t)$ divides $R(t+a)$.  The two polynomials have
the same degree, so $R(t+a)=cR(t)$ for some $c\in\F^\times$.  Similarly
$S(t+a)=cS(t)$ with the same scalar $c$.  Applying the shift twice gives $c^2=1$.  In
characteristic two, this forces $c=1$.  Hence $R$ and $S$ are individually invariant under
$t\mapsto t+a$.

Let $f(t)=pt^2+qt+r$ be any polynomial of degree at most $2$.  In characteristic two,
\[
f(t+a)-f(t)=pa^2+qa.
\]
Thus $f(t+a)=f(t)$ if and only if $q=pa$, i.e.
\[
f(t)=r+p(t^2+at).
\]
Therefore every invariant polynomial of degree at most $2$ lies in
$\Span\{1,t^2+at\}$.

If $R,S$ are linearly independent, they form a basis of this two-dimensional invariant space.
Taking pairwise products gives the full symmetric-square space
\[
\Span\{1,
\ t^2+at,
\ (t^2+at)^2\}.
\]
\end{IEEEproof}

\begin{lemma}[Three-polynomial involution lemma]
\label{lem:algebraic-projection}
Assume a certificate exists.  For every pair $i\neq k$, there is a nonidentity involution
\[
\tau_{ik}\in \Mob(\F)
\]
such that
\[
t_{kj}=\tau_{ik}(t_{ij})
\qquad\text{for every }j=1,\ldots,6.
\]
Moreover, in characteristic two, if $\tau_{ik}(t)=t+a$ with $a\neq0$, then for every
full-row-rank matrix $W\in\F^{3\times5}$ with right kernel $\Span(a_i,a_k)$, the three
coordinate polynomials of $W\nu(t)$ span
\[
\Span\{1,
\ t^2+at,
\ (t^2+at)^2\}.
\]
\end{lemma}

\begin{IEEEproof}
Fix $i\neq k$.  By Lemma~\ref{lem:left-independent},
\[
U:=\Span(a_i,a_k)
\]
is a two-dimensional subspace of $\F^5$.  Choose a full-row-rank matrix
\[
W\in \F^{3\times5}
\]
with right kernel $U$, and define
\[
q(t)=W\nu(t)=(q_1(t),q_2(t),q_3(t))^T\in\F[t]^3.
\]
The polynomials $q_1,q_2,q_3$ are linearly independent.  Indeed, a relation
$c_1q_1+c_2q_2+c_3q_3\equiv0$ would give $(c^TW)\nu(t)\equiv0$, hence $c^TW=0$ because
$1,t,t^2,t^3,t^4$ are linearly independent.  Since $W$ has row rank $3$, this gives $c=0$.

Applying $W$ to \eqref{eq:basic-certificate} gives
\[
q(t_{ij})=y_{ij}Wb_j,
\qquad
q(t_{kj})=y_{kj}Wb_j.
\]
Thus $q(t_{ij})$ and $q(t_{kj})$ are proportional for every $j$, allowing the zero vector for
now.

Let
\[
\begin{aligned}
D(t)&=\gcd(q_1,q_2,q_3),\\
q(t)&=D(t)\widetilde q(t),\\
e&:=\deg D.
\end{aligned}
\]
We first prove $e\leq2$.  The three $q_\ell$ span a $3$-dimensional subspace of
$\F[t]_{\leq4}$.  Since all are divisible by $D$, this subspace lies inside
\[
D(t)\F[t]_{\leq4-e},
\]
whose dimension is $5-e$.  Therefore $3\leq5-e$, so $e\leq2$.

The components of $\widetilde q$ have no common factor.  Since $\F$ is algebraically closed,
they have no common root.  Hence $q(t_0)=0$ if and only if $D(t_0)=0$.
Call the pair $\{t_{ij},t_{kj}\}$ ruined if one of its two values is a root of $D$.  The six
pairs are disjoint because all $36$ evaluation points are distinct.  Since $D$ has at most
$e$ distinct roots, at most $e$ pairs are ruined.  Hence at least $6-e$ pairs are intact.
Choose exactly $5-e$ intact pairs.

The vector space of homogeneous quadratic forms in three variables has dimension $6$.
The conditions
\[
Q(\widetilde q(t_{ij}))=0
\]
for the chosen $5-e$ intact pairs are homogeneous linear conditions on the coefficients of
$Q$.  Hence there is a nonzero homogeneous quadratic $Q$ satisfying these conditions.  For
each chosen intact pair, the two nonzero vectors $\widetilde q(t_{ij})$ and
$\widetilde q(t_{kj})$ are proportional; since $Q$ is homogeneous, $Q$ vanishes at both
values of the pair.  Therefore
\[
P(t):=Q(\widetilde q(t))
\]
has at least $2(5-e)=10-2e$ distinct roots.  But each coordinate of $\widetilde q$ has degree
at most $4-e$, so
\[
\deg P\leq2(4-e)=8-2e.
\]
Thus $P\equiv0$.

The quadratic $Q$ is irreducible.  If $Q=L_1L_2$ were a product of linear forms, then
\[
L_1(\widetilde q(t))L_2(\widetilde q(t))\equiv0.
\]
Since $\F[t]$ is an integral domain, one factor, say $L_1(\widetilde q(t))$, is identically
zero.  Multiplying by $D$ gives a nontrivial linear relation among $q_1,q_2,q_3$, impossible.

By Lemma~\ref{lem:quadratic-normal-form}, after an invertible linear change of the three
coordinates of $\widetilde q$, we obtain polynomials $F_1,F_2,F_3$ such that
\[
F_2(t)^2=F_1(t)F_3(t).
\]
The polynomials $F_1,F_2,F_3$ remain linearly independent, have no common factor, and have
maximum degree at most $4-e$.  Lemma~\ref{lem:square-factorization} gives coprime polynomials
$R,S$ with
\[
F_1=\lambda R^2,
\qquad
F_2=\eta RS,
\qquad
F_3=\mu S^2
\]
for nonzero constants, and
\[
2h\leq4-e,
\qquad
h:=\max(\deg R,\deg S),
\qquad
h\neq0.
\]

For every chosen intact pair, the triples
\[
(F_1(t_{ij}),F_2(t_{ij}),F_3(t_{ij}))
\quad\text{and}\quad
(F_1(t_{kj}),F_2(t_{kj}),F_3(t_{kj}))
\]
are proportional.  Using the displayed factorization, this implies that the two pairs
\[
(R(t_{ij}),S(t_{ij}))
\quad\text{and}\quad
(R(t_{kj}),S(t_{kj}))
\]
are proportional.  Indeed, for nonzero pairs $(r,s)$ and $(r',s')$, proportionality of
$(r^2,rs,s^2)$ and $(r'^2,r's',s'^2)$ gives $(rs'-r's)^2=0$, hence $rs'=r's$ in the field
$\F$.  Equivalently,
\[
R(t_{ij})S(t_{kj})-R(t_{kj})S(t_{ij})=0.
\]

If $h=1$, the equality above forces $t_{ij}=t_{kj}$ for every chosen intact pair: two
coprime linear polynomials separate distinct scalar values.  This contradicts distinctness of
the Vandermonde evaluation points.  Hence $h=2$.  The inequality $2h\leq4-e$ now forces
$e=0$.

Thus $D$ is constant, so no evaluation point is ruined.  The same coprime quadratic pair
$R,S$ satisfies
\[
R(t_{ij})S(t_{kj})-R(t_{kj})S(t_{ij})=0
\]
for all six values of $j$.  Since $t_{ij}\neq t_{kj}$, Lemma~\ref{lem:degree-two-involution}
gives a nonidentity involution $\tau_{ik}\in\Mob(\F)$ with
\[
t_{kj}=\tau_{ik}(t_{ij}),
\qquad j=1,\ldots,6.
\]

For the final assertion in characteristic two, suppose $\tau_{ik}(t)=t+a$, $a\neq0$.  The
same $R,S$ satisfy
\[
R(t+a)S(t)-R(t)S(t+a)\equiv0.
\]
By Lemma~\ref{lem:translation-invariant-quadratics},
\[
\Span\{R^2,RS,S^2\}=\Span\{1,t^2+at,(t^2+at)^2\}.
\]
The coordinate change from $\widetilde q$ to $(F_1,F_2,F_3)$ was invertible and $D$ is
constant, so the coordinate polynomials of $q(t)=W\nu(t)$ span the same space.  This proves
the characteristic-two assertion.
\end{IEEEproof}

\subsection{Fractional-linear involutions}

\begin{lemma}[Pairwise quotient involutions]
\label{lem:mob-quotients}
Let $g_1,\ldots,g_s\in \Mob(\F)$ be distinct fractional linear transformations such that
$g_k g_i^{-1}$ is an involution for every $i\neq k$.
\begin{enumerate}[label=(\roman*)]
\item If $\operatorname{char}\F\neq2$, then $s\leq4$.
\item If $\operatorname{char}\F=2$, then after one fractional linear change of variable, all
$g_i$ are translations
\[
g_i(t)=t+\alpha_i,
\]
with distinct $\alpha_i\in\F$.
\end{enumerate}
\end{lemma}

\begin{IEEEproof}
Replace each $g_i$ by $g_i g_1^{-1}$.  Then $g_1=\id$, and each $g_i$ for $i>1$ is an
involution.  For $i,k>1$, the quotient condition says $g_kg_i^{-1}=g_kg_i$ is also an
involution.  Hence
\[
(g_kg_i)^2=\id.
\]
Since $g_i^2=g_k^2=\id$, this implies $g_kg_i=g_ig_k$.  Thus the $g_i$ commute.

Assume first $\operatorname{char}\F\neq2$.  Choose a nonidentity involution among the $g_i$
and change variable so that it is $t\mapsto -t$.  Any fractional linear transformation
commuting with $t\mapsto -t$ must preserve the set $\{0,\infty\}$.  Hence it has one of the
forms
\[
t\mapsto \lambda t,
\qquad
 t\mapsto \lambda/t.
\]
Among maps $t\mapsto\lambda t$, being an involution forces $\lambda=\pm1$.  If two maps
$t\mapsto\lambda/t$ and $t\mapsto\mu/t$ both occur in a commuting elementary two-group,
their product is $t\mapsto(\lambda/\mu)t$, so $\lambda/\mu=\pm1$.  Thus there are at most
four such transformations.  Therefore $s\leq4$.

Now assume $\operatorname{char}\F=2$.  A nonidentity fractional-linear involution has a
unique fixed point in $\widehat\F$: its representing $2\times2$ matrix has a repeated
eigenvalue, and if it had two independent eigenvectors it would be scalar.  Since the
involutions commute, they share this fixed point.  Indeed, if $a$ has unique fixed point
$P$ and $b$ commutes with $a$, then $b(P)$ is also fixed by $a$, so $b(P)=P$.

Make one fractional linear change of variable sending the common fixed point to $\infty$.
Then every $g_i$ is affine:
\[
g_i(t)=\lambda_i t+\beta_i.
\]
The equation $g_i^2=\id$ gives $\lambda_i^2=1$, hence $\lambda_i=1$ in characteristic two.
Thus $g_i(t)=t+\alpha_i$.  The $\alpha_i$ are distinct because the $g_i$ are distinct.
\end{IEEEproof}

\subsection{Completion of proof}

Assume, for contradiction, that a full-rank generalized Vandermonde certificate of distance
$6$ exists.  Work over an algebraically closed field as above.  We have distinct scalars
$t_{ij}$ and vectors $a_i,b_j\in\F^5$ satisfying
\[
\nu(t_{ij})=x_{ij}a_i+y_{ij}b_j.
\]

For each $i\neq k$, Lemma~\ref{lem:algebraic-projection} gives a nonidentity involution
$\tau_{ik}\in\Mob(\F)$ such that
\[
t_{kj}=\tau_{ik}(t_{ij}),
\qquad j=1,\ldots,6.
\]
Fix the first left row and define
\[
g_1=\id,
\qquad
 g_i=\tau_{1i}\quad (i=2,\ldots,6).
\]
Then
\[
t_{ij}=g_i(t_{1j})
\qquad\text{for all }i,j.
\]
For $i\neq k$,
\[
g_k(t_{1j})=t_{kj}=\tau_{ik}(t_{ij})=\tau_{ik}g_i(t_{1j})
\]
for $j=1,\ldots,6$.  The six values $t_{11},\ldots,t_{16}$ are distinct.  A fractional
linear transformation is determined by its values on three distinct elements of $\widehat\F$.
Therefore
\[
g_k=\tau_{ik}g_i,
\qquad\text{so}\qquad
\tau_{ik}=g_kg_i^{-1}.
\]
Thus every pairwise quotient $g_kg_i^{-1}$ is an involution.  The $g_i$ are distinct,
because $g_i=g_k$ would imply $t_{ij}=t_{kj}$ for all $j$, contradicting distinctness of the
$36$ Vandermonde evaluation points.

If $\operatorname{char}\F\neq2$, Lemma~\ref{lem:mob-quotients} gives $6\leq4$, a
contradiction.  It remains to handle characteristic two.

Assume $\operatorname{char}\F=2$.  By Lemma~\ref{lem:mob-quotients}, after a fractional
linear change of the parameter $t$,
\[
g_i(t)=t+\alpha_i,
\qquad i=1,\ldots,6,
\]
with the $\alpha_i$ distinct.  The common fixed point of these translations is $\infty$.
No value $t_{1j}$ is this fixed point; otherwise all six values
$g_i(t_{1j})$ would be equal, contradicting distinctness.  Hence, after this change of
variable, all relevant evaluation points remain finite.  The induced change of basis on the
five-dimensional space spanned by $1,t,t^2,t^3,t^4$ lets us continue to write
\[
\nu(t)=(1,t,t^2,t^3,t^4)^T.
\]

For $i\neq k$, the involution
\[
\tau_{ik}=g_kg_i^{-1}
\]
is the translation
\[
t\longmapsto t+a_{ik},
\qquad
 a_{ik}:=\alpha_i+
\alpha_k\neq0.
\]
Choose any full-row-rank matrix $W_{ik}\in\F^{3\times5}$ with right kernel
$\Span(a_i,a_k)$.  By the characteristic-two part of Lemma~\ref{lem:algebraic-projection},
the coordinate polynomials of $W_{ik}\nu(t)$ span
\[
\Span\{1,
\ t^2+a_{ik}t,
\ (t^2+a_{ik}t)^2\}.
\]
Since
\[
(t^2+a_{ik}t)^2=t^4+a_{ik}^2t^2
\]
in characteristic two, every polynomial in this span has zero coefficient of $t^3$.  Therefore
every row of $W_{ik}$ has zero fourth coordinate.  Equivalently,
\[
Q_0:=(0,0,0,1,0)^T\in\ker W_{ik}=\Span(a_i,a_k)
\]
for every pair $i\neq k$.

Choose $p$ such that $a_p$ is not proportional to $Q_0$.  Such a $p$ exists because the
$a_i$ are pairwise linearly independent.  For every $k\neq p$,
\[
Q_0\in\Span(a_p,a_k).
\]
Since $a_p$ and $Q_0$ are linearly independent, this forces
\[
a_k\in\Span(a_p,Q_0).
\]
Therefore all six vectors $a_1,\ldots,a_6$ lie in one fixed two-dimensional subspace
\[
P:=\Span(a_p,Q_0)\subset\F^5.
\]

Fix any right index $j$.  From \eqref{eq:basic-certificate}, for every $i$,
\[
\nu(t_{ij})=x_{ij}a_i+y_{ij}b_j\in P+\Span(b_j).
\]
The space $P+\Span(b_j)$ has dimension at most $3$.  Hence the six distinct vectors
\[
\nu(t_{1j}),\ldots,\nu(t_{6j})
\]
all lie in a three-dimensional vector space.  This contradicts
Lemma~\ref{lem:moment-independence}, since any four distinct vectors of the form $\nu(t)$
are linearly independent.

This contradiction rules out characteristic two.  Therefore no full-rank generalized
Vandermonde certificate of distance $6$ exists over any field, proving
Theorem~\ref{thm:no-cert}.

\bibliographystyle{ieeetr}
\bibliography{ref}

@article{hallmarriage,
author = {Hall, P.},
title = {On Representatives of Subsets},
journal = {Journal of the London Mathematical Society},
volume = {s1-10},
number = {1},
pages = {26-30},
doi = {https://doi.org/10.1112/jlms/s1-10.37.26},
url = {https://londmathsoc.onlinelibrary.wiley.com/doi/abs/10.1112/jlms/s1-10.37.26},
eprint = {https://londmathsoc.onlinelibrary.wiley.com/doi/pdf/10.1112/jlms/s1-10.37.26},
month = jan,
year = {1935}
}

@article{dauMA,
author = {Son Hoang Dau and Wentu Song and Chau Yuen},
title = {On Simple Multiple Access Networks},
year = {2015},
issue_date = {Feb. 2015},
publisher = {IEEE Press},
volume = {33},
number = {2},
issn = {0733-8716},
url = {https://doi.org/10.1109/JSAC.2014.2384295},
doi = {10.1109/JSAC.2014.2384295},
journal = {IEEE J.Sel. A. Commun.},
month = feb,
pages = {236–249},
numpages = {14}
}

@INPROCEEDINGS{halbawidistribute,
  author={Halbawi, Wael and Ho, Tracey and Yao, Hongyi and Duursma, Iwan},
  booktitle={2014 IEEE International Symposium on Information Theory}, 
  title={Distributed {Reed--Solomon} codes for simple multiple access networks}, 
  year={2014},
  month=jun,
  volume={},
  number={},
  pages={651-655},
  doi={10.1109/ISIT.2014.6874913}}

@INPROCEEDINGS{DauGMMDS,
  author={Dau, Son Hoang and Song, Wentu and Yuen, Chau},
  booktitle={2014 IEEE International Symposium on Information Theory}, 
  title={On the existence of {MDS} codes over small fields with constrained generator matrices}, 
  year={2014},
  month=jun,
  volume={},
  number={},
  pages={1787-1791},
  doi={10.1109/ISIT.2014.6875141}}

@ARTICLE{YildizGMMDS,
  author={Yildiz, Hikmet and Hassibi, Babak},
  journal={IEEE Transactions on Information Theory}, 
  title={Optimum Linear Codes With Support-Constrained Generator Matrices Over Small Fields}, 
  year={2019},
  month=dec,
  volume={65},
  number={12},
  pages={7868-7875},
  doi={10.1109/TIT.2019.2932663}}

@INPROCEEDINGS{LovettGMMDS,
  author={Lovett, Shachar},
  booktitle={2018 IEEE 59th Annual Symposium on Foundations of Computer Science (FOCS)}, 
  title={{MDS} matrices over Small Fields: A Proof of the {GM-MDS} Conjecture}, 
  year={2018},
  month=oct,
  volume={},
  number={},
  pages={194-199},
  doi={10.1109/FOCS.2018.00027}}

@ARTICLE{LDPCCap,
  author={Richardson, T.J. and Shokrollahi, M.A. and Urbanke, R.L.},
  journal={IEEE Transactions on Information Theory}, 
  title={Design of capacity-approaching irregular low-density parity-check codes}, 
  year={2001},
  month=feb,
  volume={47},
  number={2},
  pages={619-637},
  doi={10.1109/18.910578}}

@ARTICLE{GallagerLDPC,
  author={Gallager, R.},
  journal={IRE Transactions on Information Theory}, 
  title={Low-density parity-check codes}, 
  year={1962},
  month=jan,
  volume={8},
  number={1},
  pages={21-28},
  doi={10.1109/TIT.1962.1057683}}

@article{FowlerSurface,
   title={Surface codes: Towards practical large-scale quantum computation},
   volume={86},
   ISSN={1094-1622},
   url={http://dx.doi.org/10.1103/PhysRevA.86.032324},
   DOI={10.1103/physreva.86.032324},
   number={3},
   journal={Physical Review A},
   publisher={American Physical Society (APS)},
   author={Fowler, Austin G. and Mariantoni, Matteo and Martinis, John M. and Cleland, Andrew N.},
   year={2012},
   month=sep}

@misc{KitaevSurface,
      title={Quantum codes on a lattice with boundary}, 
      author={S. B. Bravyi and A. Yu. Kitaev},
      year={1998},
      eprint={quant-ph/9811052},
      archivePrefix={arXiv},
      primaryClass={quant-ph},
      url={https://arxiv.org/abs/quant-ph/9811052}, 
}

@article{gottesmanLDPC,
author = {Gottesman, Daniel},
title = {Fault-tolerant quantum computation with constant overhead},
year = {2014},
issue_date = {November 2014},
publisher = {Rinton Press, Incorporated},
address = {Paramus, NJ},
volume = {14},
number = {15–16},
issn = {1533-7146},
journal = {Quantum Info. Comput.},
month = nov,
pages = {1338–1372},
numpages = {35}
}

@ARTICLE{TamoLRC,
  author={Sharma, Sandeep and Ramkumar, Vinayak and Tamo, Itzhak},
  journal={IEEE Journal on Selected Areas in Information Theory}, 
  title={Quantum Locally Recoverable Codes via Good Polynomials}, 
  year={2025},
  month=may,
  volume={6},
  number={},
  pages={100-110},
  doi={10.1109/JSAIT.2025.3567480}}

@inproceedings{GolowichLRC,
  author    = {Louis Golowich and Venkatesan Guruswami},
  title     = {Quantum Locally Recoverable Codes},
  booktitle = {Proceedings of the 2025 Annual ACM-SIAM Symposium on Discrete Algorithms (SODA)},
  year      = {2025},
  month     = jan,
  publisher = {Society for Industrial and Applied Mathematics (SIAM)},
  address   = {Philadelphia, PA},
  pages     = {5512--5522},
  doi       = {10.1137/1.9781611978322.188},
  url       = {https://epubs.siam.org/doi/abs/10.1137/1.9781611978322.188}
}

@article{CalderbankCSS,
   title={Good quantum error-correcting codes exist},
   volume={54},
   ISSN={1094-1622},
   url={http://dx.doi.org/10.1103/PhysRevA.54.1098},
   DOI={10.1103/physreva.54.1098},
   number={2},
   journal={Physical Review A},
   publisher={American Physical Society (APS)},
   author={Calderbank, A. R. and Shor, Peter W.},
   year={1996},
   month=aug, pages={1098–1105} }

@article{SteaneCSS,
author = {Steane, Andrew},
year = {1996},
month = nov,
pages = {},
title = {Multiple Particle Interference and Quantum Error Correction},
volume = {452},
journal = {Proceedings of the Royal Society A: Mathematical, Physical and Engineering Sciences},
doi = {10.1098/rspa.1996.0136}
}

@inbook{GrasslRS,
   title={Quantum Reed--Solomon Codes},
   ISBN={9783540467960},
   ISSN={0302-9743},
   url={http://dx.doi.org/10.1007/3-540-46796-3_23},
   DOI={10.1007/3-540-46796-3_23},
   booktitle={Applied Algebra, Algebraic Algorithms and Error-Correcting Codes},
   publisher={Springer Berlin Heidelberg},
   author={Grassl, Markus and Geiselmann, Willi and Beth, Thomas},
   year={1999},
   pages={231–244} }

@ARTICLE{halbawisparse,
  author={Halbawi, Wael and Liu, Zihan and Duursma, Iwan M. and Dau, Hoang and Hassibi, Babak},
  journal={IEEE Transactions on Information Theory}, 
  title={Sparse and Balanced {Reed--Solomon} and {Tamo–Barg Codes}}, 
  year={2019},
  month=jan,
  volume={65},
  number={1},
  pages={118-130},
  doi={10.1109/TIT.2018.2873128}}

@article{RuizQuantum,
   title={{LDPC-cat} codes for low-overhead quantum computing in 2D},
   volume={16},
   ISSN={2041-1723},
   url={http://dx.doi.org/10.1038/s41467-025-56298-8},
   DOI={10.1038/s41467-025-56298-8},
   number={1},
   journal={Nature Communications},
   publisher={Springer Science and Business Media LLC},
   author={Ruiz, Diego and Guillaud, Jérémie and Leverrier, Anthony and Mirrahimi, Mazyar and Vuillot, Christophe},
   year={2025},
   month=jan }

@article{SaboQuantum,
   title={Weight-Reduced Stabilizer Codes with Lower Overhead},
   volume={5},
   ISSN={2691-3399},
   url={http://dx.doi.org/10.1103/PRXQuantum.5.040302},
   DOI={10.1103/prxquantum.5.040302},
   number={4},
   journal={PRX Quantum},
   publisher={American Physical Society (APS)},
   author={Sabo, Eric and Gunderman, Lane G. and Ide, Benjamin and Vasmer, Michael and Dauphinais, Guillaume},
   year={2024},
   month=oct }

@ARTICLE{TamoBargLRC,
  author={Tamo, Itzhak and Barg, Alexander},
  journal={IEEE Transactions on Information Theory},
  title={A Family of Optimal Locally Recoverable Codes},
  year={2014},
  month=aug,
  volume={60},
  number={8},
  pages={4661-4676},
  doi={10.1109/TIT.2014.2321280}}

@misc{3GPP5G,
  author={{3rd Generation Partnership Project}},
  title={{NR}; Multiplexing and channel coding ({3GPP} {TS} 38.212)},
  year={2018},
  note={Release 15, v15.2.0}}

@INPROCEEDINGS{SipserSpielman,
  author={Sipser, M. and Spielman, D.A.},
  booktitle={Proceedings of the Twenty-Sixth Annual ACM Symposium on Theory of Computing},
  title={Expander codes},
  year={1996},
  month=may,
  pages={261-271},
  doi={10.1145/237814.237985}}

@ARTICLE{BlaumPMDS,
  author={Blaum, Mario and Hafner, James Lee and Hetzler, Steven},
  journal={IEEE Transactions on Information Theory},
  title={Partial-{MDS} Codes and Their Application to {RAID} Type of Architectures},
  year={2013},
  month=jul,
  volume={59},
  number={7},
  pages={4510--4519},
  doi={10.1109/TIT.2013.2252395}}

\end{document}